\ifx\CCCG\undefined%
\documentclass[12pt]{article}%
\newcommand{\CCCGVer}[1]{}%
\newcommand{\RegVer}[1]{#1}%
\else
\documentclass{styles/cccg23}
\newcommand{\CCCGVer}[1]{#1}%
\newcommand{\RegVer}[1]{}%
\fi

\RegVer{%
   \def\UseBibLatex{1}%
}

\ifx\sarielComp\undefined%
\newcommand{\SarielComp}[1]{}
\newcommand{\NotSarielComp}[1]{#1}%
\else
\newcommand{\SarielComp}[1]{#1}%
\newcommand{\NotSarielComp}[1]{}%
\fi

\newcommand{\UsePackage}[1]{%
  \IfFileExists{styles/#1.sty}{%
      \usepackage{styles/#1}%
   }{%
      \IfFileExists{../styles/#1.sty}{%
         \usepackage{../styles/#1}%
      }{%
         \usepackage{#1}%
      }%
   }%
}

\RegVer{%
   \usepackage[cm]{fullpage}%
}%
\usepackage{amsmath}%
\usepackage{amssymb}%
\usepackage{xcolor}%

\SarielComp{\usepackage{sariel_colors}}%

\RegVer{%
   \usepackage[amsmath,thmmarks]{ntheorem}%
   \theoremseparator{.}%

   \usepackage{titlesec}%
   \titlelabel{\thetitle. }%
}

\usepackage{xcolor}%
\usepackage{mleftright}%
\usepackage{xspace}%
\usepackage{hyperref}%

\usepackage{hyperref}%
\hypersetup{%
      unicode,
      breaklinks,%
      colorlinks=true,%
      urlcolor=[rgb]{0.25,0.0,0.0},%
      linkcolor=[rgb]{0.5,0.0,0.0},%
      citecolor=[rgb]{0,0.2,0.445},%
      filecolor=[rgb]{0,0,0.4},
      anchorcolor=[rgb]={0.0,0.1,0.2}%
}
\usepackage[ocgcolorlinks]{ocgx2}

\usepackage[algo2e,boxed,linesnumbered,noend]{algorithm2e}
\usepackage[noend]{algpseudocode}

\CCCGVer{%
   \newtheorem{problem}[theorem]{Problem}
   \newtheorem{corollary}[theorem]{Corollary}%
}%
\RegVer{%
   \theoremseparator{.}%

\theoremstyle{plain}%
\newtheorem{theorem}{Theorem}[section]

\newtheorem{lemma}[theorem]{Lemma}

\newtheorem{corollary}[theorem]{Corollary}

\theoremstyle{plain}%
\theoremheaderfont{\sf} \theorembodyfont{\upshape}%
\newtheorem*{remark:unnumbered}[theorem]{Remark}%

\newtheorem{problem}[theorem]{Problem}

\newcommand{\myqedsymbol}{\rule{2mm}{2mm}}

\theoremheaderfont{\em}%
\theorembodyfont{\upshape}%
\theoremstyle{nonumberplain}%
\theoremseparator{}%
\theoremsymbol{\myqedsymbol}%
\newtheorem{proof}{Proof:}%

}%

\definecolor{blue25emph}{rgb}{0, 0, 11}
\providecommand{\emphic}[2]{%
   \textcolor{blue25emph}{%
      \textbf{\emph{#1}}}%
   \index{#2}}

\providecommand{\emphi}[1]{\emphic{#1}{#1}}

\definecolor{almostblack}{rgb}{0, 0, 0.3}
\providecommand{\emphw}[1]{{\textcolor{almostblack}{\emph{#1}}}}%

\newcommand{\atgen}{\symbol{'100}}
\newcommand{\SarielThanks}[1]{\thanks{Department of Computer Science;
      University of Illinois; 201 N. Goodwin Avenue; Urbana, IL,
      61801, USA; {\tt sariel\atgen{}illinois.edu}; {\tt
         \url{http://sarielhp.org/}.} #1}}

\newcommand{\BenThanks}[1]{\thanks{Department of Computer Science;
      University of Texas at Dallas; Richardson, TX 75080, USA;
      {\tt benjamin.raichel\atgen{}utdallas.edu}; {\tt
         \url{http://utdallas.edu/\string~benjamin.raichel}.} #1}}

\newcommand{\HLink}[2]{\hyperref[#2]{#1~\ref*{#2}}}
\newcommand{\HLinkSuffix}[3]{\hyperref[#2]{#1\ref*{#2}{#3}}}

\newcommand{\thmlab}[1]{{\label{theo:#1}}}
\newcommand{\thmref}[1]{\HLink{Theorem}{theo:#1}}

\newcommand{\corlab}[1]{\label{cor:#1}}
\newcommand{\corref}[1]{\HLink{Corollary}{cor:#1}}%

\newcommand{\problab}[1]{\label{problem:#1}}

\newcommand{\lemlab}[1]{\label{lemma:#1}}
\newcommand{\lemref}[1]{\HLink{Lemma}{lemma:#1}}%

\providecommand{\eqlab}[1]{}%
\renewcommand{\eqlab}[1]{\label{equation:#1}}
\newcommand{\Eqref}[1]{\HLinkSuffix{Eq.~(}{equation:#1}{)}}

\newcommand{\remove}[1]{}%
\newcommand{\Set}[2]{\left\{ #1 \;\middle\vert\; #2 \right\}}
\newcommand{\pth}[2][\!]{\mleft({#2}\mright)}%

\newcommand{\cardin}[1]{\lvert {#1} \rvert}%

\renewcommand{\Re}{\mathbb{R}}%

\usepackage[inline]{enumitem}

\newlist{compactenumA}{enumerate}{5}%
\setlist[compactenumA]{topsep=0pt,itemsep=-1ex,partopsep=1ex,parsep=1ex,%
   label=(\Alph*)}%

\newlist{compactenuma}{enumerate}{5}%
\setlist[compactenuma]{topsep=0pt,itemsep=-1ex,partopsep=1ex,parsep=1ex,%
   label=(\alph*)}%

\newlist{compactenumI}{enumerate}{5}%
\setlist[compactenumI]{topsep=0pt,itemsep=-1ex,partopsep=1ex,parsep=1ex,%
   label=(\Roman*)}%

\newlist{compactenumi}{enumerate}{5}%
\setlist[compactenumi]{topsep=0pt,itemsep=-1ex,partopsep=1ex,parsep=1ex,%
   label=(\roman*)}%

\newlist{compactitem}{itemize}{5}%
\setlist[compactitem]{topsep=0pt,itemsep=-1ex,partopsep=1ex,parsep=1ex,%
   label=\ensuremath{\bullet}}%

\providecommand{\BibLatexMode}[1]{}
\providecommand{\BibTexMode}[1]{#1}

\ifx\UseBibLatex\undefined%
  \renewcommand{\BibLatexMode}[1]{}
  \renewcommand{\BibTexMode}[1]{#1}
\else
  \renewcommand{\BibLatexMode}[1]{#1}
  \renewcommand{\BibTexMode}[1]{}
\fi

\BibLatexMode{%
   \usepackage[bibencoding=utf8,style=alphabetic,backend=biber]{biblatex}%
   \UsePackage{sariel_biblatex}%
}

\providecommand{\Mh}[1]{#1}%

\newcommand{\CHX}[1]{{\mathcal{C}}\pth{#1}}%
\newcommand{\dsY}[2]{\mathsf{d}\pth{#1,#2}}

\newcommand{\DHY}[2]{\mathsf{D}_H\pth{#1,#2}}

\newcommand{\dsHY}[2]{\mathsf{d}_H\pth{#1,#2}}
\newcommand{\dsHDY}[2]{\mathsf{d}\pth{#1 \rightarrow #2}}

\newcommand{\normX}[1]{\lVert #1 \rVert}%
\newcommand{\dY}[2]{\normX{#1 #2}}%
\renewcommand{\Re}{\mathbb{R}}%
\newcommand{\eps}{{\varepsilon}}%

\newcommand{\etal}{\textit{et~al.}\xspace}

\DefineNamedColor{named}{RedViolet} {cmyk}{0.07,0.90,0,0.34}

\newcommand{\AlgorithmI}[1]{{\textcolor[named]{RedViolet}{\texttt{\bf{#1}}}}}
\newcommand{\Algorithm}[1]{{\AlgorithmI{#1}\index{algorithm!#1@{\AlgorithmI{#1}}}}}
\newcommand{\decider}{\Algorithm{decider}\xspace}

\newcommand{\extremal}{\Algorithm{extremal}\xspace}

\providecommand{\P}{} \renewcommand{\P}{\Mh{P}}%

\providecommand{\Q}{}
\renewcommand{\Q}{\Mh{Q}}%

\providecommand{\CS}{}
\renewcommand{\CS}{\Mh{Q}}%

\newcommand{\optX}[1]{#1^\star}

\newcommand{\Qopt}{\Mh{\optX{Q}}}%

\newcommand{\MinCardin}{\textsf{Min{}Cardin}}%
\newcommand{\refMinCardin}{\hyperref[problem:min:k]{\MinCardin}\xspace}

\newcommand{\MinDist}{\textsf{Min{}Dist}}%
\newcommand{\refMinDist}{\hyperref[problem:min:dist]{\MinDist}\xspace}

\newcommand{\rad}{\Mh{\tau}}%

\newcommand{\ropt}{\optX{\rad}}%

\newcommand{\roptY}[2]{\optX{\rad}\pth{#1,#2}}%

\newcommand{\Family}{\Mh{\mathcal{F}}}%
\newcommand{\FeqX}[1]{\Family_{\leq #1}}
\newcommand{\VX}[1]{\mathsf{V}\pth{#1}}

\newcommand{\kopt}{\Mh{\optX{k}}}%
\newcommand{\koptY}[2]{\Mh{\optX{k}}\pth{#1,#2}}%

\newcommand{\RunningTime}{O\bigl(n^{3/2} \sqrt{k} \log^{3/2} n + kn
   \log^2 n\bigr)}%

\BibLatexMode{%
   \bibliography{h_min_eps}%
}

\begin{document}

\title{On the Budgeted Hausdorff Distance Problem}

\author{%
   Sariel Har-Peled%
   \SarielThanks{Work on this paper was partially supported by a NSF
      AF award CCF-1907400.  }%
   \and%
   Benjamin Raichel%
   \BenThanks{Work on this paper was partially supported by NSF CAREER
      Award 1750780.}  }

\date{}%

\maketitle
\begin{abstract}
    Given a set $\P$ of $n$ points in the plane, and a parameter $k$,
    we present an algorithm, whose running time is $\RunningTime$,
    with high probability, that computes a subset $\Qopt \subseteq \P$
    of $k$ points, that minimizes the Hausdorff distance between the
    convex-hulls of $\Qopt$ and $\P$. This is the first subquadratic
    algorithm for this problem if $k$ is small.
\end{abstract}

\section{Introduction}

Given a set of points $\P$ in $\Re^d$, a natural goal is to find a
small subset of it that represents the point set well. This problem
has attracted a lot of interest over the last two decades, and this
subset of $\P$ is usually referred to as a \emph{coreset}
\cite{ahv-aemp-04, ahv-gavc-05}. An alternative approximation is
provided by the largest enclosed ellipsoid inside $\CHX{\P}$ (here
$\CHX{\P}$ denotes the \emph{convex-hull} of $\P$) or the smallest
area bounding box of $\P$ (not necessarily axis-aligned). This provides a
constant approximation to the projection width of $\P$ in any
direction $v$ -- that is, the projection of $\P$ into the line spanned
by $v$ is contained in the projection of the ellipsoid after
appropriate constant scaling.  One can show that in two dimensions,
there is a subset $\CS \subseteq \P$ (i.e., a coreset) of size
$O(1/\sqrt{\eps})$ such that the projection width of $\P$ and $\CS$ is
the same up to scaling by $1+\eps$. See Agarwal \etal
\cite{ahv-aemp-04, ahv-gavc-05} for more details.

The concept of a coreset is attractive as it provides a notion of
approximating that adapts to the shape of the point set. However, an
older and arguably simpler approach is to require that $\CHX{\CS}$
approximates $\CHX{\P}$ within a certain absolute error threshold. A
natural such measure is the \emphi{Hausdorff} distance between sets
$X,Y \subseteq \Re^2$, which is
\begin{equation}
    \dsHY{X}{Y}%
    =%
    \max\pth{ \dsHDY{X}{Y},\,  \dsHDY{Y}{X}\Bigr.},
    \qquad
\end{equation}
where
\begin{equation*}
    \dsHDY{X}{Y} = \max_{x \in X} \min_{y \in Y}\dY{x}{y}
    \eqlab{d:h}.
\end{equation*}

In our specific case, the two sets are $\CHX{\P}$ and $\CHX{\CS}$, and
let $\DHY{\CS}{\P} = \dsHY{\CHX{\CS}}{\CHX{\P} \bigr.}$. The natural
questions are
\begin{compactenumI}
    \medskip%
    \item \refMinCardin: Compute the smallest subset
    $\CS \subseteq \P$, such that $\DHY{\CS}{\P} \leq \rad$, where
    $\rad$ is a prespecified error threshold. Formally, let
    \begin{equation*}
        \FeqX{\rad} = \Set{\CS \subseteq \P}{\DHY{\CS}{\P} \leq \rad
           \bigr.},
    \end{equation*}
    and let
    $\kopt = \koptY{\P}{\rad} = \min_{\CS \in \FeqX{\rad}}
    \cardin{\CS}$ denote the minimum cardinality of such a set $\CS$.

    \medskip%
    \item \refMinDist: Compute the subset $\CS \subseteq \P$ of size
    $k$, such that $\DHY{\CS}{\P}$ is minimized, where $k$ is a
    prespecified subset size threshold. Let
    $\ropt = \roptY{\P}{k} = \min_{\CS \subseteq \P: \cardin{\CS}= k}
    \DHY{\CS}{\P}$ denote the optimal radius.
\end{compactenumI}
\medskip%
The two problems are ``dual'' to each other -- solve one, and you get
a way to solve the other in polynomial time via a search on the values
of the other parameter. In particular, solving both problems directly
(in two dimensions) can be done via dynamic programming, but even
getting a subcubic running time is not immediate in this case.
Indeed, the problem seems to have a surprisingly subtle and intricate
structure that make this problem more challenging than it seems at first.

Klimenko and Raichel \cite{kr-fechs-21} provided an $O(n^{2.53})$ time
algorithm for \refMinCardin. Very recently, Agarwal and Har-Peled \cite{ah-coktd-23}
provided a near-linear time algorithm for \refMinCardin that runs in
near linear time if $\kopt = \koptY{\P}{\rad}$ is small. Specifically,
the running time of this algorithm is $O( \kopt n \log n)$.

The purpose of this work is to come up with a subquadratic algorithm
for the ``dual'' problem \refMinDist. An algorithm with running time
$O(n^2 \log n)$ follows readily by computing all possible critical
values, and performing a binary search over these values, using the
procedure of \cite{ah-coktd-23} as a black box.  The only subquadratic
algorithm known previously was for the special case when $\P$ is in
convex position, for which \cite{kr-fechs-21} gave an algorithm whose
running time is $O(n\log^3 n)$ with high probability.

Our main result is an algorithm that, given $\P$ and $k$ as input,
solves \refMinDist in $\RunningTime$ time, with high probability, see
\thmref{main:res} for details. We believe the algorithm itself is
technically interesting -- it uses random sampling to reduce the range
of interest into an interval containing $O( \sqrt{n})$ critical
values. It then use the decision procedure of \cite{ah-coktd-23} as a
way to compute the critical values in this interval, by ``peeling''
them one by one in decreasing order. Using random sampling for
parametric search is an old idea, see \cite{hr-fdre-14} and references
there.

\section{Preliminaries}

Given a point set $X$ in $\Re^2$, let $\CHX{X}$ denote its
\emphi{convex hull}.  For two compact sets $X,Y\subset \Re^2$, let
$\dsY{X}{Y} = \min_{x\in X, y\in Y} \dY{x}{y}$ denote their
distance. For a single point $x$ let $\dsY{x}{Y} = \dsY{\{x\}}{Y}$.

Consider two finite point sets $\Q \subseteq \P\subset \Re^2$, and
observe that
\[
    \DHY{\Q}{\P}%
    =%
    \dsHY{\CHX{\CS}}{\CHX{\P} \bigr.}
    =%
    \max_{p \in \P} \dsY{p}{\CHX{\Q}\bigr.},
\]
see \Eqref{d:h}.
The first equality above is by definition, and the second is since $Q\subseteq P$ and so we have that
$\CHX{Q}\subseteq \CHX{P}$, and moreover the furthest point in
$\CHX{P}$ from $\CHX{Q}$ is always a point in $P$.

In this paper we consider the following two related problems, where
for simplicity, we assume that $P$ is in general position.

\begin{problem}[{\MinCardin}]%
    \problab{min:k}%
    Given a set $\P \subset \Re^2$ of $n$ points, and a value
    $\rad > 0$, find the smallest cardinality subset $\CS\subseteq \P$
    such that $\DHY{\CS}{\P}\leq \rad$. %
\end{problem}

\begin{problem}[\MinDist]%
    \problab{min:dist}
    Given a set $\P \subset \Re^2$ of $n$ points, and an integer $k$,
    find the subset $\Q\subseteq \P$ that minimizes $\DHY{\Q}{\P}$
    subject to the constraint that $\cardin{\Q}\leq k$.
\end{problem}

For either problem let $\Qopt$ denote an optimal solution. For
\refMinCardin let $\kopt=\koptY{ \P}{ \rad}=\cardin{\Qopt}$, and for
\refMinDist let $\ropt=\roptY{\P}{k}=\DHY{\Qopt}{\P}$.  The algorithms
discussed in this paper will output the set $\Qopt$, though when it
eases the exposition, we occasionally refer to $\kopt$ as the solution
to \refMinCardin and $\ropt$ as the solution to \refMinDist.

\begin{theorem}[\cite{ah-coktd-23}]%
    \thmlab{decider}%
    Given as an input a point set $\P$ and parameters $k$ and $\rad$,
    let $\kopt = \kopt(\P,\rad)$.  There is a procedure
    \decider{}$( \P, \rad, k)$, that in $O(nk\log n)$ time, either
    returns that ``$\kopt > k$'', or alternatively returns a set
    $\Qopt \subseteq \P$, such that $\cardin{\Qopt} =\kopt \leq k$,
    and $\DHY{\Qopt}{\P} \leq \rad$.
\end{theorem}
The above theorem readily implies that the problem \refMinCardin can
be solved in $O(n\kopt\log n)$ time.

Given an input of size $n$, an algorithm runs in $O(f(n))$ time
\emphw{with high probability}, if for any chosen constant $c>0$, there
is a constant $\alpha_c$ such that the running time exceeds
$\alpha_c f(n)$ with probability $<1/n^c$.

\section{Algorithm}

\newcommand{\CE}{\Mh{\Xi}}%
\newcommand{\LE}{\Mh{\mathcal{L}}}%
\newcommand{\VE}{\Mh{\mathcal{V}}}%

\subsection{The canonical set}

Given an instance $P,k$ of \refMinDist, let $\Qopt$ denote an optimal
solution. Recall that
\begin{equation*}
    \ropt = \DHY{\Qopt}{\P} = \max_{p \in P} \dsY{p}{\CHX{\Qopt}}.
\end{equation*}
Assume that $\ropt>0$, which can easily be determined by checking if
$\cardin{\VX{\CHX{P}}}> k$, where $\VX{\CHX{P}}$ denotes the set of
vertices of $\CHX{P}$.  Let
\begin{equation*}
    p = \arg\max_{p' \in P} \dsY{p'}{\CHX{\Qopt}},
\end{equation*}
and let $q$ be its projection onto $\CHX{\Qopt}$, i.e.\
$\ropt=\dY{p}{q}$.  Observe that $q$ either lies on a vertex of
$\CHX{\Qopt}$ or in the interior of a bounding edge. Since
$\Qopt\subseteq \P$, we can conclude that $\ropt$ is either (i) the
distance between two points in $\P$, or (ii) the distance from a point
in $\P$ to the line passing through two other points from $\P$. Note
that, in case (ii), $q$ must be the orthogonal projection of $p$ on to
the line $\ell$ supporting the edge, and that $p$ must be the furthest
point from $\ell$ out of the points that lie in one of its two
defining halfplanes. In particular, for an ordered pair $a,b\in \P$
define $\ell_{a,b}$ as the line through $a$ and $b$, directed from $a$
to $b$, and let $\P_{a,b}$ be the subset of $\P$ lying in the
halfspace bounded by and to the left of $\ell_{a,b}$. We thus define
the following two sets.
\begin{align}
  \CCCGVer{&}
             \VE = \Set{\bigl.\dY{x}{y}}{ x,y\in \P}%
             \CCCGVer{\nonumber\\}%
  \RegVer{\qquad}\text{and}\RegVer{\qquad}%
  \CCCGVer{\qquad&}
                   \CCCGVer{\nonumber\\&}%
  \LE = \Set{\bigl.\smash{\max_{p\in \P_{a,b}}}
  \dsY{p}{\ell_{a,b}}}{ a,b\in \P}.
  \eqlab{L:E}
\end{align}
The set $\CE=\VE \cup \LE$ is the \emphi{canonical set} of distance
values (i.e., the set of all \emph{critical} values).  By the above
discussion, we have $\ropt\in \CE$.

Observe that $\VE$ and $\LE$ (and hence $\CE$) have quadratic
size. Thus we will not explicitly compute these sets. Instead we will
search over $\VE$ using the following ``median'' selection
procedure.

\begin{theorem}[\cite{cz-hplsfc-21}]%
    \thmlab{selection}%
    Given a set $\P\subset \Re^2$ of $n$ points, and an integer $k>0$,
    with high probability, in $O(n^{4/3})$ time, one can compute the
    value of rank $k$ in $\VE$.
\end{theorem}

For values in $\LE$, the algorithm samples values and searches over them,
using a procedure loosely inspired by \cite{hr-fdre-14}. For that we
have the following standard lemma, whose proof we include for
completeness.

\begin{lemma}[\cite{dk-fdpi-83}]
    \lemlab{extreme}%
    Let $\P\subset \Re^2$ be a set of $n$ points.  Then in
    $O(n\log n)$ time one can build a data structure such that for any
    query vector $\overrightarrow{u}$, in $O(\log n)$ time, it returns
    the point of $P$ extremal in the direction $\overrightarrow{u}$,
    i.e.\ the point maximizing the dot product with
    $\overrightarrow{u}$.  Let $\extremal(\overrightarrow{u})$ denote
    this query procedure.
\end{lemma}
\begin{proof}
    Let $\VX{\CHX{P}}=\{q_1,\ldots,q_k\}$ be labelled in clockwise
    order. Let $U(q_i)$ be the set of unit vectors
    $\overrightarrow{u}$ such that when we translate $P$ so that $q_i$
    lies at the origin, then $\overrightarrow{u}$ lies in the exterior
    angle between the normals of $q_{i-1}q_i$ and $q_iq_{i+1}$. Observe that
    $\extremal(\overrightarrow{u})=q_i$ precisely when $u\in
    U(q_i)$. Moreover, the $U(q_i)$ define a partition of the set of
    all unit vectors into $k$ sets. Thus if we maintain these
    intervals in a array, sorted in clockwise order, then in
    $O(\log k)=O(\log n)$ time we can binary search to find which
    interval $\overrightarrow{u}$ falls in. It takes $O(n\log n)$ time
    to compute $\CHX{P}$ and thus the data structure.
\end{proof}

In the next section, given a directed line $\ell$, we use the above lemma to make extremal
queries for the normal of $\ell$ lying in its left defining
halfplane. This lets us evaluate extreme points for lines supporting edges of the current hull, as well as allows us to sample values from $\LE$, for which we have the following.

\begin{corollary}
    \corlab{samp}%
    Given a set $P\subset \Re^2$ of $n$ points, after $O(n\log n)$
    preprocessing time, one can return, in $O(\log n)$ time, a value
    sampled uniformly at random from $\LE$.
\end{corollary}
\begin{proof}
    Sample uniformly at random a pair of points from $\P$, and then
    use \lemref{extreme} for the normal to the line passing through
    this pair of points.
\end{proof}

\subsection{The algorithm in stages}

The input is a set $\P$ of $n$ points, and a parameter $k$. The task
at hand is to compute the minimum distance $\ropt$, such that there is
a subset $\Q \subseteq \P$ of size $k$, such that
$\DHY{\Q}{\P}\leq \ropt$.

\paragraph{Searching and testing for the optimal value.}

The algorithm maintains an interval $(r,R)$, such that the following invariants are maintained:
\begin{compactenumI}
    \item $\kopt(\P,r) >k$,
    \smallskip%
    \item $\kopt(\P,R) \leq k$, and
    \smallskip%
    \item $\ropt(\P,k) \in (r,R)$.
\end{compactenumI}
\smallskip%
(The first two conditions are actually implied by the last condition,
though for clarity we list all three.)  In the following, let
$\delta>0$ denote an infinitesimal\footnote{The algorithm can be
   describe without using infinitesimals, but this is somewhat
   cleaner.}. Given a value $\rad \in (r,R)$, one can decide if
$\rad=\ropt(\P,k)$, by running $\decider(\P,\rad, k)$ and
$\decider(\P,\rad -\delta, k)$, see \thmref{decider}. If
$\decider(\P,\rad -\delta, k)$ returns that
$\kopt( \P, \rad-\delta) >k$ and $\kopt( \P, \rad) =k$ then clearly
$\rad$ is the desired optimal value. In this case, the algorithm
returns this value and stops.

\paragraph{Updating the current interval.}

After testing if $\rad=\ropt(\P,k)$ for a value $\rad \in (r,R)$ as described above, if $\rad\neq \ropt(\P,k)$ then the algorithm
can update the current interval. Indeed, if
$\decider(\P,\rad, k)$ returns that $\kopt( \P, \rad) >k$, then the algorithm sets the current
interval to $(\rad,R)$. Otherwise, $\decider(\P,\rad-\delta, k)$ returned that $\kopt( \P, \rad-\delta) \leq k$ and so the algorithm sets the current interval to $(r,\rad)$.

\paragraph{Stage I: Handling pairwise distances.}

The algorithm sets the initial interval to $(0, \infty)$. (Recall as
discussed above that we can assume $\ropt>0$.) The algorithm then
binary searches over all pairwise distance from $\VE = \binom{\P}{2}$
by using the distance selection procedure of \thmref{selection}, in
the process repeatedly updating the current interval as described
above. If $\ropt\in \VE$, then the algorithm will terminate when the
search considers this value. Otherwise, this search reduces the
current interval to two consecutive pairwise distances from $\VE$,
$r<R$, such that $\ropt\in (r,R)$ and the current interval $(r,R)$
contains no pairwise distance of $\P$ in its interior.

\paragraph{Stage II: Sampling edge-vertex  distances.}
The algorithm samples a set $\Pi$ of $O(n^{3/2}\log n)$ values from
$\LE$, see \Eqref{L:E}, using \corref{samp}. Let $U$ be the subset of
values of $\Pi$ that lie inside the current interval. The algorithm
binary searches over $U$, repeatedly updating the current interval as
described above (by doing median selection so that $U$'s cardinality
halves at each iteration).  If $\ropt\in U$ then the algorithm will
terminate when the search considers this value. Otherwise, the search
further reduces to the interval to $I'=(r',R')$.  (Which as discussed
below, with high probability, contains $O(\sqrt{n})$ values from
$\LE$.)

\paragraph{Stage III: Peeling the critical  edge-vertex  distances.}
The algorithm now continues the search on the interval $I'= (r',R')$
and critical values in it, $I' \cap \CE = I' \cap \LE$. In particular,
the solution computed by $\decider(\P,R', k)$ is a set
$\CS \subseteq \P$ of size $\leq k$ such that $\DHY{\CS}{\P} \leq R'$.
For every edge on the boundary of $\CHX{Q}$ the algorithm now computes
the point from $\P$ furthest away from the line supporting the edge
(among the points in the halfplane not containing $\CHX{Q}$), using
extremal queries from \lemref{extreme}.  Let $\alpha$ be the largest
such computed value over all the edges, and observe that
$\alpha = \DHY{\CS}{\P}$.\footnote{$\DHY{\CS}{\P}$ must be realized at
   a value from $\LE$ as Stage I eliminated $\VE$ values, and thus it
   sufficed to consider furthest distances to the lines supporting
   edges rather than the edges themselves, since at the maximum such
   value they must align.}
If $\alpha < R'$, then $\alpha \geq \ropt(\P,k)$.
The algorithm tests if $\alpha= \ropt(\P,k)$, and if so it terminates.
Otherwise, it must be that the optimal value lies in the interval $(r',\alpha)$.
As $\alpha \in (r', R')$ and $\alpha\in \LE$, our new interval $(r',\alpha)$ has at least one fewer value from $\LE$.
The algorithm now continues to the next iteration of Stage III.

The case when $\alpha =R'$ (i.e., the higher end of the active
interval) is somewhat more subtle. The algorithm calls
$\decider(\P, k, \alpha - \delta)$ to compute a set $\CS'$ that
realizes $\kopt(\P,\alpha-\delta)$, where $\delta$ is an
infinitesimal. Observe that $\kopt(\P,\alpha-\delta)\leq k$, as
otherwise $\alpha=R'$ was the desired optimal value. Let
$\beta = \DHY{\CS'}{\P}$, which can be computed in a similar fashion
using $Q'$ as $\alpha$ was computed using $Q$.
The algorithm tests if $\beta= \ropt(\P,k)$, and if so it terminates.
Otherwise, by the same reasoning used above for $\alpha$, we can conclude our new interval $(r',\beta)$ has at least one fewer value from $\LE$, and thus the algorithm continues to the next iteration of Stage III on the interval $(r',\beta)$.

\subsection{Analysis}

\paragraph{Correctness.}
The correctness of the algorithm is fairly immediate given the
discussion above. Namely, the algorithm maintains an interval $(r,R)$
with the invariant that $\ropt(\P,k)\in (r,R)$ (where initially this
interval is $(0,\infty)$). In each step of each stage a value
$\rad\in (r,R)$ that is either from $\VE$ (in Stage I) or from $\LE$
(in Stages II and III) is determined. For this value $\rad$ we then
update the current interval as described above. Namely, we query
$\decider(\P,\rad,k)$ and $\decider(\P,\rad-\delta,k)$.
If these calls return that  $\kopt( \P, \rad) \leq k$ and $\kopt( \P, \rad-\delta) > k$ then $\rad=\ropt(\P,k)$ and the algorithm terminates.
Otherwise, if $\kopt( \P, \rad) >k$ the algorithm proceeds on $(\rad,R)$, and if $\kopt( \P, \rad-\delta) \leq k$ then it proceeds on $(r,\rad)$.
In either case the interval contains at least one fewer value from $\CE$, and thus eventually the algorithm must terminate with the value $\ropt(\P,k)$.

\paragraph{Running time analysis.}
In Stage I the algorithm performs a binary search over
$\VE = \binom{\P}{2}$. This is done using the distance selection
procedure of \thmref{selection}, which with high probability takes
$O(n^{4/3})$ time to determine each next query value. Each query is
answered using the $O(nk\log n)$ time $\decider(\P,\cdot ,k)$ from
\thmref{decider}. Thus in total Stage I takes
$O\bigl((n^{4/3}+nk\log n)\log n\bigr)$ time with high
probability. Here, by the union bound, a polynomial number of high
probability events (i.e.\ the events that each call to selection
occurs in $O(n^{4/3})$ time), all occur simultaneously with high
probability.

In Stage II the algorithm samples $O(n^{3/2}\log n)$ values from $\LE$
using the $O(\log n)$ time sampling procedure of \corref{samp}. Next,
the algorithm binary searches over these values (this time directly),
again using $\decider(\P,\cdot ,k)$. Thus in total Stage II takes
$O(n^{3/2}\log^2 n + nk\log^2 n)$ time.

Stage III begins with some interval $(r',R')$. Let
$X=\cardin{\LE\cap (r',R')}$.  In each iteration of Stage III, for
some subset $\Q\subseteq \P$ of size at most $k$, the algorithm
computes $\alpha = \DHY{\CS}{\P}$. This is done using at most $k$
calls to the $O(\log n)$ query time \lemref{extreme}. (This same step
is potentially done a second time for $\beta = \DHY{\CS'}{\P}$). Each
iteration of Stage III also performs a constant number of calls to
$\decider(P,\cdot,k)$, thus is total one iteration takes
$O(k\log n+nk\log n) = O(nk\log n)$ time. As argued above each
iteration of Stage III reduces the number of values from $\LE$ in the
active interval by at least 1, and thus runs for at most $X$
iterations. Thus the total time of Stage III is $O(Xnk\log n)$.

Observe that since Stage II sampled a set $\Pi$ of $O(n^{3/2}\log n)$ values from the $O(n^2)$ sized set $\LE$, the interval between any two consecutive values of $\Pi$ with high probability has $O(\sqrt{n})$ values from $\LE$. As the interval $I'=(r',R')$ returned by Stage II is such an interval, with high probability $X=O(\sqrt{n})$. As the running time of Stage II dominates the running time of Stage I (with high probability), we thus have that with high probability the total time of all stages is
\begin{align*}
  &O(n^{3/2}\log^2 n+(\log n + X)nk\log n)
    \CCCGVer{\\&}%
  =
  O(n^{3/2}\log^2 n + n^{3/2} k \log n + nk \log^2 n)%
  \CCCGVer{\\&}%
  =%
  O(n^{3/2} (k + \log n) \log n ).
\end{align*}

\paragraph{Slightly improving the running time.} %
Observe that if the algorithm samples $O(n t \log n)$ values in stage
II, then with high probability the last two stages take
\begin{equation*}
    O\pth{ n t \log^2 n + \pth{\frac{n^2}{nt} + \log n} k n \log n}
\end{equation*}
time. Solving for $t$, we have
\begin{equation*}
    n t \log^2 n = (n^2/t) k  \log n
    \implies%
    t^2 = n k/\log n.
\end{equation*}
Thus, setting $t=\sqrt{nk / \log n }$, and including the running time of stage I, we get the improved high probability running time bound
\begin{align*}
  &O\pth{ n^{4/3} \log n +
    n t \log^2 n + \pth{\frac{n^2}{nt} + \log n} k n \log n}
  \\&
  =%
  \RunningTime.
\end{align*}
In summary, we get the following result.

\begin{theorem}
    \thmlab{main:res}%
    Given an instance of \refMinDist, consisting of a set
    $P\subset \Re^2$ of $n$ points and an integer $k$, the above
    algorithm computes a set $\Qopt \subseteq \P$, of size $k$, that
    realizes the minimum Hausdorff distance between the convex-hulls
    of $\P$ and $\Qopt$ among all such subsets -- that is,
    $\roptY{\P}{k} = \DHY{\P}{\Qopt}$. The running time of the
    algorithm is $\RunningTime$ with high probability.
\end{theorem}

We remark that under the reasonable assumption that $k=O(n/\log n)$ the running time can be stated more simply as $O(n^{3/2}\sqrt{k}\log^{3/2} n)$.

\section{Conclusions}

The most interesting open problem left by our work is whether one can
get a near-linear running time if $k$ is small. Even beating
$O(n^{4/3})$ seems challenging. On the other hand, if one is willing to use $2k$ points
then a near linear running time is achievable \cite{kr-fechs-21}. However, using less than $2k$ points without increasing the Hausdorff distance in near linear
time seems challenging.

\RegVer{%
   \BibTexMode{%
      \bibliographystyle{alpha}%
      \bibliography{h_min_eps}%
   }%
}
\CCCGVer{%
   \small
   \bibliographystyle{abbrv}%
   \bibliography{h_min_eps}%
}

\BibLatexMode{\printbibliography}

\end{document}